\documentclass[11pt]{amsart}
\usepackage[english]{babel}
\usepackage{amssymb,amsthm,amsmath}
\usepackage[numbers,sort&compress]{natbib}
\usepackage{color}
\usepackage{graphicx}
\usepackage{tikz}
\usepackage[colorlinks,
            linkcolor=blue,
            anchorcolor=blue,
            citecolor=blue
            ]{hyperref}

\numberwithin{equation}{section}

\def\R{\mathbb{R}}

\def\Z{\mathbb{Z}}

\def\E{\mathbb{E}}

\newcommand{\elz}{\ell^2 (\Z^d)}

\newcommand{\elr}{L^2 (\R^d)}
\newcommand{\supp}{\mathrm{supp\,}}

\newcommand{\blm}{\mathrm{d_{BL}}(\mu, \tilde\mu)}
\newcommand{\blr}{\mathrm{d_{BL}}(\rho, \tilde\rho)}

\newcommand{\krm}{\mathrm{d_{KR}}(\mu, \tilde\mu)}
\newcommand{\krr}{\mathrm{d_{KR}}(\rho, \tilde\rho)}

\newcommand{\tds}{\tilde{\rho}}
\newcommand{\tids}{\tilde{N}}
\newcommand{\hrestr}{H_\Lambda}

\def\tr{{\text{tr}}}

\theoremstyle{plain}

\newtheorem*{theorem*}{Theorem}
\newtheorem{nono-theorem}{Theorem}[]

\newtheorem{theorem}{Theorem}
\theoremstyle{plain}
\newtheorem{proposition}{Proposition}[section]

\newtheorem{rmk}[proposition]{Remark} 

\newtheorem{claim}[proposition]{Claim}

\begin{document}
\title[Continuity of the integrated density of states]{On the continuity of the integrated density of states in the disorder}

\author {Mira Shamis}
\address{School of Mathematical Sciences,
Queen Mary University of London, 
Mile End Road, London E1 4NS, England} \email{m.shamis@qmul.ac.uk}

\maketitle

\begin{abstract}
Recently, Hislop and Marx studied the dependence of the integrated density of states on the underlying probability distribution for a class of discrete random Schr\"odinger operators, and established a quantitative form of continuity in weak* topology. We develop an alternative approach to the problem, based on Ky Fan inequalities, and establish a sharp version of the estimate of Hislop and Marx. We also consider a corresponding problem for continual random Schr\"odinger operators on $\R^d$.
\end{abstract}

\section{Introduction}

Recently, Hislop and Marx \cite{hm} studied the dependence of the integrated density of states (IDS) of random Schr\"odinger operators on the distribution of the potential.

\smallskip

\noindent Let $\{ V(n)\}_{n\in\Z^d}$ be independent identically distributed random variables (i.i.d.r.v.) with the common probability distribution $\mu$. Let $H$ be the random Schr\"odinger operator acting on $\elz$ by

\begin{equation}\label{def-opH}
H = -\Delta + V,\quad (H\psi)(n) = \sum_{m\,\, \text{is adjacent to\,\,} n} (\psi(n) - \psi(m)) + V(n)\psi(n).~
\end{equation}
The IDS corresponding to the operator $H$ is the function

\begin{equation}\label{ids-disc}
N(E) = \lim_{\Lambda\nearrow\Z^d}\frac{1}{|\Lambda|}\,\#\,\{\text{eigenvalues of\,} \hrestr \, \text{in}\,  (-\infty, E]\},~
\end{equation}
where $\hrestr$ is the restriction of $H$ to a finite box $\Lambda\subset\Z^d$,\ i.e.\ $\hrestr = P_\Lambda H P_\Lambda^*$, where $P_\Lambda: \elz \rightarrow \ell^2(\Lambda)$ is the coordinate projection ((\ref{ids-disc}) holds with probability $1$). The measure with cumulative distribution function $N$ is denoted by $\rho$.

\smallskip

\noindent To discuss the dependence of $\rho$ on the distribution of the potential $\mu$, we introduce two metrics on the space of Borel probability measures on $\R$.

\smallskip

\noindent The Kantorovich-Rubinstein (Wasserstein) metric  is defined via

\begin{equation}\label{def-kr}
\krm  = \sup\left\{ \left|\int f\, d\mu - \int f\, d\tilde\mu\right|\, :\,\, f: \R\rightarrow\, \R\, \text{is 1-Lipschitz}\,\right\}.~
\end{equation}

By the Kantorovich-Rubinstein duality theorem
\begin{equation}\label{def-kr-exp}
\krm =  \inf\{ \E |X - \tilde X|\},~
\end{equation}
where the infimum is taken over $\R^2$-valued random variables $(X, \tilde X)$, such that $X\sim\mu, \tilde X\sim \tilde\mu$.
Following \cite{hm}, we also consider the bounded Lipschitz metric, defined by
\begin{equation}\label{def-bl}
\blm = \sup\left\{ \left|\int f\, d\mu - \int f\, d\tilde\mu\right|\, :\, f: \R\rightarrow\, \R\, \text{is 1-Lip},\, \|f\|_\infty \leq 1\right\}.~
\end{equation}
Our definition differs from \cite{hm} by a multiplicative constant. Observe that 
\begin{equation*}
\blm \leq \krm,~
\end{equation*}
and if $\supp\mu, \supp\tilde\mu \subset [-A, A]$, then
\begin{equation*}
 \krm \leq \max (A, 1) \blm.~
\end{equation*}
In this notation Theorem 1.1 of \cite{hm} (formulated here in slightly less general setting than in the cited work) asserts the following.

\begin{theorem*}[Hislop--Marx] Suppose $H, \tilde H$ are random Schr\"odinger operators of the form (\ref{def-opH}) with potentials $\{ V(n)\}, \{ \tilde V(n)\}$ sampled from a probability distributions $\mu, \tilde\mu$, respectively. Denote by $N, \tids$ the IDS corresponding to $H, \tilde H$, and let $\rho, \tds$ be the measures with cumulative distribution functions $N, \tids$, respectively. If  $\supp\mu, \supp\tilde\mu \subset [-A, A]$, then
\begin{equation}\label{hm-dos}
\blr \leq C_A \blm^{1/(1 + 2d)},~
\end{equation}
\begin{equation}\label{hm-ids}
\sup|N(E) - \tids(E)|\leq\frac{C_A}{\log_+{\frac{1}{\blm}}},~
\end{equation}
where $C_A$ depends only on $A$.
\end{theorem*}

We refer to \cite{hm} for a discussion of earlier work on the subject, and only mention the result \cite{hks} on the continuity of the integrated density of states as a function of the coupling constant.

\medskip

Hislop and Marx \cite{hm} presented several applications, particularly, to the continuity of the Lyapunov exponent of a one-dimensional operator as a function of the underlying distribution of the potential. The proof of the theorem in \cite{hm} is based on the approximation of the function $f$ (in (\ref{def-bl})) by polynomials. 

\smallskip

We suggest a different approach to estimates of the form (\ref{hm-dos}) using the Ky Fan inequalities. Our first result is the following theorem.

\begin{theorem}\label{th:disc} Suppose $H, \tilde H$ are random Schr\"odinger operators of the form (\ref{def-opH}), where $\{ V(n)\}, \{ \tilde V(n)\}$ are i.i.d.r.v. distributed accordingly to $\mu, \tilde\mu$ respectively. Let $N, \tids$ be the IDS corresponding to $H, \tilde H$, and let $\rho, \tds$ be the measures with cumulative distribution functions $N, \tids$ respectively. Then,
\begin{equation}\label{dos}
\krr \leq \krm,~
\end{equation}
\begin{equation}\label{ids}
\sup_E|N(E) - \tids(E)|\leq\frac{C}{\log_+{\frac{1}{\krm}}},~
\end{equation}
where $C > 0$ is a numerical constant.
\end{theorem}

\begin{rmk} The power $1$ as well as the prefactor $1$ in (\ref{dos}) are optimal in general.
\end{rmk}

\begin{rmk} This result can be extended to other models in which the potential is of the form
\begin{equation}\label{finite-rank}
\sum v_j P_j,~
\end{equation}
where $v_j$ are i.i.d.r.v. with common Borel distribution supported on a finite interval and $P_j$ are finite rank projections (see \cite{hm}).
\end{rmk}
\begin{rmk} Theorem \ref{th:disc} can be extended to different underlying lattices, since the proof does not rely on the structure of $\Z^d$.
\end{rmk}

\medskip

\noindent In the follow up paper \cite{hm2}, Hislop and Marx prove a version of their results for the continual Anderson model, which is \textit{not} of the form (\ref{finite-rank}). A modification of our argument can be applied to the continual setting as well. We illustrate it by the following theorem.

\medskip

\noindent Let $H$ be a random Schr\"odinger operator acting on $L^2(\R^d)$, defined by

\begin{equation}\label{def-opHcont}
H = -\Delta + V,~
\end{equation}
where the potential $V$ is of the form
\begin{equation}\label{def-v}
V(x) = \sum_{j\in\Z^d} v_j u(x - j),\,\ x\in\R^d,
\end{equation}
where $v_j$ are i.i.d.r.v. distributed accordingly to $\mu$, and $u$ is real-valued continuous compactly supported function: $u\in C_c(\R)$. Denote by $\Lambda$ the cube of side length $L$ around the origin
\[
\Lambda = \left[-\frac{L}{2}, \frac{L}{2}\right]^d.~
\]
Let $\hrestr$ be the restriction of $H$ to  $L^2(\Lambda)$ with Dirichlet boundary conditions. Define the IDS corresponding to $H$ similarly to (\ref{ids-disc})
\begin{equation}\label{eq:ids-cont}
N(E) = \lim_{L\to\infty}\frac{1}{L^d}\,\#\,\{\text{eigenvalues of\,} \hrestr \, \text{in}\,  (-\infty, E]\},~
\end{equation}
and let $\rho$ be the measure with cumulative distribution function $N$.

\begin{theorem}\label{th:cont} Suppose $H, \tilde H$ are random Schr\"odinger operators of the form (\ref{def-opHcont}), and suppose that $\supp\,\mu, \supp\,\tilde\mu\subset\R_+$ and $u\geq 0$. Let $N, \tids$ be the IDS corresponding to $H, \tilde H$, and let $\rho, \tds$ be the measures with cumulative distribution functions $N, \tids$ respectively. If $\alpha > \frac{d}{2} - 1$, then
\begin{equation}\label{dos-cont}
\left|\int f\left(\frac{1}{(1 + E)^\alpha}\right) d\rho(E) - \int f\left(\frac{1}{(1 + E)^\alpha}\right) d\tds(E)\right|\leq C(d, u, \alpha) \krm,~
\end{equation}
for any $1$-Lipschitz function $f$ for which $\int f\left(\frac{1}{1 + E}\right) d\rho(E)$ converges.

If $d = 1,2,3$ and $\supp\, \mu\subset [0, A]$ then for any $E_0\in\R$
\begin{equation}\label{ids-cont}
\sup_{E\leq E_0}|N(E) - \tids(E)|\leq\frac{C(d, E_0, A)}{\log_+^{\kappa_d}{\frac{1}{\krm}}},\,\,~
\end{equation}
where $\kappa_1 = 1, \kappa_2 = 1/4, \kappa_3 = 1/8$.
\end{theorem}

\begin{rmk} The following example shows that the condition $\alpha > \frac{d}{2} - 1$ is optimal in general, and in particular one can not expect a result of the same form as in the discrete case (which would correspond to $\alpha = -1$).

Assume $\alpha \leq \frac{d}{2} - 1$. Let $u$ be such that $\sum_{j\in\Z^d} u(x - j) \equiv 1$, and let $v_j \equiv 0, \tilde v_j \equiv \delta, f_\epsilon(x) = \max((x - \epsilon), 0)$. The integrated density of states of the free Lapacian is given by
\[
d\rho (\lambda) = C_d \lambda^{\frac{d}{2} - 1}d\lambda.
\]
Therefore we have
\begin{equation}\label{eq:one}
\int f((1 + \lambda)^{-\alpha}) d\rho(\lambda) = C_d\int_0^{\epsilon^{-1/\alpha} - 1} ((1 + \lambda)^{-\alpha} - \epsilon) \lambda^{\frac{d}{2} - 1}d\lambda,
\end{equation}
\begin{equation}\label{eq:two}
\int f((1 + \lambda)^{-\alpha}) d\tilde\rho(\lambda) = C_d\int_0^{\epsilon^{-1/\alpha} - 1 - \delta} ((1 + \lambda + \delta)^{-\alpha} - \epsilon) \lambda^{\frac{d}{2} - 1}d\lambda,
\end{equation}
and for any $\delta > 0$ we obtain
\[
\begin{split}
&\liminf_{\epsilon\to 0} \left|\int f((1 + \lambda)^{-\alpha}) d\rho(\lambda) - \int f((1 + \lambda)^{-\alpha}) d\tilde\rho(\lambda)\right|  \\ 
& \geq \liminf_{\epsilon\to 0} C_d \int_0^{\epsilon^{-1/\alpha} - 1 - \delta} ((1 + \lambda)^{-\alpha} - (1 + \lambda + \delta)^{-\alpha}) \lambda^{\frac{d}{2} - 1}d\lambda  \\
&\geq \liminf_{\epsilon\to 0} C_d\alpha\delta \int_0^{\epsilon^{-1/\alpha} - 1 - \delta} (1 + \lambda + \delta)^{-\alpha - 1} \lambda^{\frac{d}{2} - 1}d\lambda = \infty.
\end{split}
\]
\end{rmk}

\begin{rmk} The restrictions on the dimension and on $V$ in the second part of Theorem~\ref{th:cont} come from the work of Bourgain and Klein \cite{bk} which we use to deduce (\ref{ids-cont}) from (\ref{dos-cont}).
\end{rmk} 
\begin{rmk} The restriction $\supp\, \mu\subset [0, A]$, also coming from \cite{bk}, can be relaxed using the work of Klein and Tsang \cite[Theorem 1.3]{kt}.

\end{rmk}
\begin{rmk} Theorem~\ref{th:cont} formally implies a similar result for sign-indefinite $V$ bounded from below. \end{rmk}

\section{Preliminaries}\label{sec:prel}

\subsection{Discrete case} The main ingredient of the proof of Theorem~\ref{th:disc} is the Ky Fan inequality \cite{li}:

Assume that $A, B$, and $\tilde A = A+B$ are linear self-adjoint operators that act on $n$-dimensional Euclidean space. Let $\lambda_n\leq\lambda_{n - 1}\leq\cdots\leq\lambda_1,\, e_n\leq e_{n - 1}\leq\cdots\leq e_1,\ \tilde\lambda_n \leq\tilde\lambda_{n - 1}\leq\cdots\leq\tilde\lambda_1$ be the eigenvalues of $A, B$, and $\tilde A$ respectively. Then, for any continuous convex function $\phi: \R \rightarrow \R$ 
\begin{equation}\label{eq:conv}
\sum_{j = 1}^n \phi (\tilde\lambda_j - \lambda_j) \leq \sum_{j = 1}^n \phi (e_j).~
\end{equation}
In particular,

\begin{equation}\label{eq:mod}
\sum_{j = 1}^n |\tilde\lambda_j - \lambda_j| \leq \sum_{j = 1}^n |e_j|.~
\end{equation}
To deduce (\ref{ids}) from (\ref{dos}) (similarly to \cite{hm}) we shall use the following result due to Craig and Simon \cite{cs}.
Denote by
\begin{equation}\label{mod-cont}
\omega(\delta) = \sup\, \{ |\rho(E) - \rho(E')|\, : \,\, E' < E \leq E + \delta\},~
\end{equation}
 the modulus of continuity of $\rho$. Then (\cite{cs}) the measure $\rho$ with the cumulative distribution function $N$ (the IDS) of any ergodic Schr\"odinger operator on $\elz$ is $\log$-H\"older continuous, namely, for any $\delta \in (0, \frac{1}{2}]$
\begin{equation}\label{log-holder}
\omega(\delta)\leq\frac{C}{\log\frac{1}{\delta}},~
\end{equation}
where $C > 0$ is a universal constant.

\subsection{Continual case} First, recall that for $1 \leq p < \infty$ the Schatten class $S_p$ is the class of all compact operators in a given Hilbert space such that
\[
\|A\|_p =\left( \sum_{n=1}^\infty s_n(A)^p  \right)^{1/p} < \infty,~
\]
where $\{ s_n(A)\}$ is the sequence of all singular values of the operator $A$ enumerated with multiplicities taken into account. The class $S_\infty$ consists of all compact operators.

\smallskip

The main ingredient in the proof of Theorem~\ref{th:cont} is the following version of the Ky Fan inequality (see Markus \cite{ma}).

If $A\in S_1$, $B\in S_\infty$ that are self-adjoint, and $\tilde A = A+B$,  $\lambda_1\geq\lambda_2\geq\cdots ,\, e_1\geq e_2\geq\cdots ,\, \tilde\lambda_1 \geq\tilde\lambda_2\geq\cdots$, are the eigenvalues of $A, B$, and $\tilde A$ respectively, then, for any continuous convex function $\phi: \R \rightarrow \R$ with $\phi(0) = 0$
\begin{equation}\label{eq:conv-infty}
\sum_{j = 1}^\infty \phi (\tilde\lambda_j - \lambda_j) \leq \sum_{j = 1}^\infty \phi (e_j).~
\end{equation}
In particular,

\begin{equation}\label{eq:mod-infty}
\sum_{j = 1}^\infty |\tilde\lambda_j - \lambda_j| \leq \sum_{j = 1}^\infty |e_j| = \|B\|_1.~
\end{equation}

To deduce (\ref{ids-cont}) from (\ref{dos-cont}) we will need the following result due to Bourgain and Klein \cite{bk}.

\begin{theorem*}[BK] Assume that $H$ as in (\ref{def-opHcont})--(\ref{def-v}) on $\elr$, $d = 1, 2, 3$, with $\supp\, \mu \subset [-A, A]$. Let $N$ be the corresponding IDS. Then, given $E_0\in\R$, for all $E\leq E_0$ and $\delta \leq 1/2$
\begin{equation}\label{eq:bk}
|N(E) - N(E + \delta)| \leq \frac{C(d, E_0, A)}{\log^{\kappa_d}\frac{1}{\delta}},~
\end{equation}
where  $C(d, E_0, A) > 0$, and $\kappa_1 = 1, \kappa_2 = 1/4, \kappa_3 = 1/8$.~
\end{theorem*}

\section{Proof of Theorem~\ref{th:disc} and Theorem~\ref{th:cont}}
\subsection{Proof of Theorem~\ref{th:disc}} Denote by $\Lambda\subset\Z^d$ a finite box and let (in the notation of Ky Fan's inequality)
\[
A = H_\Lambda = (-\Delta + V)_\Lambda,\,\, \tilde A = \tilde H_\Lambda =  (-\Delta + \tilde V)_\Lambda,
\]
be the restrictions of the operators $H$ and $\tilde H$ to the box $\Lambda$. Then, 

\begin{equation}\label{eq:est-trf}
\begin{split}
|\tr f(A) - \tr f(\tilde A)| & = |\sum_{j=1}^{|\Lambda|}f(\lambda_j) - \sum_{j=1}^{|\Lambda|}f(\tilde\lambda_j)| \\&\leq 
\sum_{j=1}^{|\Lambda|} |f(\lambda_j) - f(\tilde\lambda_j)| \leq  \sum_{j=1}^{|\Lambda|} |\lambda_j - \tilde\lambda_j|\\&\leq
\sum_{j=1}^{|\Lambda|} |e_j| = \sum_{x\in\Lambda} |V(x) - \tilde V(x)|,~
\end{split}
\end{equation}
where the second inequality holds since $f$ is $1$-Lipschitz and the last inequality follows from  (\ref{eq:conv}).

\smallskip

By (\ref{def-kr-exp}) there is a realization of $V$ and $\tilde V$ on a common probability space such that
\[
\E|V(x) - \tilde V(x)|\leq\krm.~
\]
Thus, using (\ref{eq:est-trf}) for any $1$-Lipschitz function $f$, we obtain
\begin{equation}\label{eq:est-tr-exp-f}
|\E\,\tr\, f(A) - \E\,\tr\, f(\tilde A)| \leq \E \sum_{x\in\Lambda} |V(x) - \tilde V(x)| \leq |\Lambda| \krm.~
\end{equation}
Since
\[
\int f d\rho = \lim_{\Lambda\nearrow\Z^d}\frac{1}{|\Lambda|}\E\,\tr\, f(A),~
\]
we obtain by passing to the limit $\Lambda\nearrow\Z^d$
\begin{equation}\label{eq:final-est-rho}
\krr\leq\krm,~
\end{equation}
thus concluding the proof of (\ref{dos}).

\smallskip
 
To deduce (\ref{ids}), we choose

\begin{equation}\label{eq:f-ids}
f(x) = \begin{cases}
\delta, &x\leq E\\
-x + E + \delta, & E \leq x \leq E + \delta\\
0, &x \geq E + \delta,~
\end{cases}
\end{equation}
for $\delta > 0$. Then, by definition of the IDS, we get for any $E\in\R$

\begin{equation}\label{eq:est-N}
\delta N(E) \leq \int f(E) d\rho (E) \leq \delta N(E + \delta),~
\end{equation}

\begin{equation}\label{eq:est-tN}
\delta \tilde N(E) \leq \int f(E) d\tilde\rho (E) \leq \delta \tilde N(E + \delta).~
\end{equation}
Since
\[
 \int f(E) d\tilde\rho (E) = \int f(E) d\rho (E) + \int f(E) d(\tilde\rho - \rho) (E),~
\]
combining (\ref{eq:final-est-rho}), (\ref{eq:est-N}), and (\ref{eq:est-tN}), we obtain
\begin{equation*}
\delta \tilde N(E) \leq \delta N(E + \delta) + \krm,~
\end{equation*}
namely
\begin{equation}\label{eq:upper-bound}
\tilde N(E) \leq  N(E + \delta) + \frac{\krm}{\delta}.~
\end{equation}
In the same way we get

\begin{equation}\label{eq:lower-bound}
\tilde N(E) \geq  N(E - \delta) - \frac{\krm}{\delta}.~
\end{equation}

Let $\omega$ be the modulus of continuity of $N$. Combining (\ref{log-holder}), (\ref{eq:upper-bound}), and (\ref{eq:lower-bound}), we obtain
\begin{equation}\label{eq:final-ids}
\sup_E |N(E) - \tilde N(E)| \leq \inf_\delta \left(\omega(\delta) + \frac{\krm}{\delta}\right) \leq \frac{C}{\log_+\frac{1}{\krm}},~
\end{equation}
where $C > 0$ is a constant and we choose $\delta = \krm/\omega(\krm)$.~ This finishes the proof of (\ref{ids}).

\qed

\begin{rmk} If the operator $H$ is such that the modulus of continuity $\omega$ satisfies
\[
\omega(\delta) \leq C\delta^a,~
\]
for some $C, a > 0$, then (\ref{eq:final-ids}) implies that
\[
\sup_E |N(E) - \tilde N(E)| \leq \inf_\delta \left( C\delta^a + \frac{\krm}{\delta} \right) \leq \tilde C\, \krm^{1/(1 + a)}.~
\] 
\end{rmk}

\subsection{Proof of Theorem~\ref{th:cont}} Let
\[
 H_\Lambda = (-\Delta + V)_\Lambda,\,\, \tilde H_\Lambda =  (-\Delta + \tilde V)_\Lambda,
\]
be the restrictions of the operators $H$ and $\tilde H$ to a finite box $\Lambda\in\R^d$ with Dirichlet boundary conditions. Let (in the notation of Ky Fan's inequality)

\[
A =  (H_\Lambda + 1)^{-\alpha},\,\, \tilde A = (\tilde H_\Lambda + 1)^{-\alpha}.~
\]
We have the following
\begin{claim}\label{cl:pf-thm2} If $\alpha > \frac{d}{2} - 1$, then
\[
\| A - \tilde A\|_1 \leq C(d, u, \alpha)\sum_{j\in 2\Lambda \cap\Z^d} |v_j - \tilde v_j|~.
\]
\end{claim}
\begin{proof}
Let us number
\[
2\Lambda \cap\Z^d = \{ j_1, \dots, j_n\},\,\,\, n\leq C|\Lambda|~,
\]
and let $H^{(k)}_\Lambda,\, 0\leq k \leq n+1$, be the (restricted) operator corresponding to the potential
\[
V(x) = \sum_{l < k} \tilde v_{j_l}u(x - j_l) + \sum_{l \geq k} v_{j_l}u(x - j_l)~.
\]
Observe that
\begin{equation}\label{eq:cl}
\begin{split}
&\| (H^{(k)}_\Lambda + 1)^{-\alpha} - (H^{(k + 1)}_\Lambda + 1)^{-\alpha}\|_1 \leq \\& |\alpha|\left(\left\| \frac{1}{(H^{(k)}_\Lambda + 1)^{\alpha + 1}} (H^{(k)}_\Lambda - H^{(k + 1)}_\Lambda)\right\|_1  + \left\| (H^{(k)}_\Lambda - H^{(k + 1)}_\Lambda) \frac{1}{(H^{(k + 1)}_\Lambda + 1)^{\alpha + 1}} \right\|_1\right)~,
\end{split}
\end{equation}
as follows, for example, from the Birman-Solomyak formula \cite[Theorem 8.1]{bs} (note that for $\alpha = 1$ it suffices to use the second resolvent identity).
Then we have
\[
\begin{split}
\left\|  \frac{1}{(H^{(k)}_\Lambda + 1)^{\alpha + 1}}u_{j_k}\right\|_1  = &\left\| \frac{1}{(H^{(k)}_\Lambda + 1)^{\frac{\alpha + 1}{2}}}\sqrt{u_{j_k}}\sqrt{u_{j_k}}\frac{1}{(H^{(k)}_\Lambda + 1)^{\frac{\alpha + 1}{2}}}  \right\|_1 \\= &  \left\|  \sqrt{u_{j_k}}\frac{1}{(H^{(k)}_\Lambda + 1)^{\frac{\alpha + 1}{2}}}\right\|_2^2 \\ = & \left\| {u_{j_k}}^{1/4}\frac{1}{(H^{(k)}_\Lambda + 1)^{\frac{\alpha + 1}{2}}}{u_{j_k}}^{1/4}\right\|_2^2 \\  \leq & \left\| {u_{j_k}}^{1/4}\frac{1}{(-\Delta_\Lambda + 1)^{\frac{\alpha + 1}{2}}}{u_{j_k}}^{1/4}\right\|_2^2  \\ \leq & \|u\|_\infty \left\| \mathbf{1}_{\mathrm{supp}\,\,u}\frac{1}{(-\Delta_\Lambda + 1)^{\frac{\alpha + 1}{2}}}\mathbf{1}_{\mathrm{supp}\,\,u}\right\|_2^2~, 
\end{split}
\]
where the first inequality follows from the positivity of the potential. A similar bound holds for the second term of (\ref{eq:cl}). By Weyl's law the last norm is bounded by $C(d, u, \alpha)$ (uniformly in $\Lambda$) whenever $\alpha + 1 > \frac{d}{2}$~.
Thus, we obtain
\[
\| (H^{(k)}_\Lambda + 1)^{-\alpha} - (H^{(k + 1)}_\Lambda + 1)^{-\alpha}\|_1 \leq C(d, u, \alpha)|v_{j_k} - \tilde v_{j_k}|.
\]

\end{proof}

%
By the Kantorovich-Rubinstein duality (\ref{def-kr-exp}) there is a realization of $v$ and $\tilde v$ on a common probability space such that
\[
\E|v_j - \tilde v_j|\leq\krm.~
\]
Thus, using Claim~\ref{cl:pf-thm2}, we get for $\alpha > \frac{d}{2} - 1$ 

\[
\begin{split}
\E \| A - \tilde A\|_1 & \leq C(d, u, \alpha)\E\sum_{j\in 2\Lambda \cap\Z^d} |v_j - \tilde v_j|  =  C(d, u, \alpha) \sum_{j\in 2\Lambda\cap\Z^d} \E  | \tilde v_j - v_j| \\&
\leq C(d, u, \alpha) |\Lambda|\krm.~
\end{split}
\]


The eigenvalues of $A$ are exactly $\frac{1}{(1 + \lambda_j)^\alpha}$, where $\lambda_j$ are the eigenvalues of $H_\Lambda$, thus using (\ref{eq:mod-infty}), we obtain for $\alpha > \frac{d}{2} - 1$ and for any $1$-Lipschitz function $f$ for which $\int f\left(\frac{1}{(1 + E)^\alpha}\right) d\rho(E)$ converges (in particular, $f(0) = 0$)

\begin{equation}\label{eq:final-step-dos-cont}
\begin{split}
&\left| \sum_{j=1}^\infty f\left(\frac{1}{(1 + \lambda_j)^\alpha}\right) - \sum_{j=1}^\infty f\left(\frac{1}{(1 + \tilde\lambda_j)^\alpha}\right) \right| \\& \leq  \sum_{j=1}^\infty \left|  f\left(\frac{1}{(1 + \lambda_j)^\alpha}\right) -  f\left(\frac{1}{(1 + \tilde\lambda_j)^\alpha}\right)\right| \\& \leq \sum_{j=1}^\infty \left|\frac{1}{(1 + \lambda_j)^\alpha} - \frac{1}{(1 +\tilde \lambda_j)^\alpha} \right| \leq C(u, d, \alpha)\, |\Lambda| \, \krm,~
\end{split}
\end{equation}
where the last step follows from the Ky Fan inequality. Using the definition (\ref{eq:ids-cont}) of $\rho$ and passing to the limit $\Lambda\nearrow\R^d$, we conclude that if $\alpha > \frac{d}{2} - 1$, then
\begin{equation}\label{eq:dos-cont}
\left|\int f\left(\frac{1}{(1 + E)^\alpha}\right) d\rho(E) - \int f\left(\frac{1}{(1 + E)^\alpha}\right) d\tds(E)\right|\leq C(d, u, \alpha) \krm~.
\end{equation} 
Thus we complete the proof of (\ref{dos-cont}).

To deduce (\ref{ids-cont}), we define

\begin{equation}\label{eq:f-ids-cont}
f(x) = \begin{cases}
\frac{\delta}{2(1 + E)^2}, &x\geq \frac{1}{1 + E}\\
\frac{1 + E + \delta}{2(1 + E)}x - \frac{1}{2(1 + E)}, & \frac{1}{1 + E + \delta} \leq x \leq \frac{1}{1 + E}\\
0, &x \leq \frac{1}{1 + E + \delta},~
\end{cases}
\end{equation}
for $\delta > 0$. Then, by the definition of the IDS  (\ref{eq:ids-cont}) we get for a fixed $E_0\in \R$ in the same way as in the proof of Theorem~\ref{th:disc}

\begin{equation}\label{eq:isd-cont-upper}
\tilde N(E_0) \leq N(E_0 + \delta) + \frac{2\krm\, (1 + E_0)^2}{\delta},~
\end{equation}

\begin{equation}\label{eq:isd-cont-lower}
\tilde N(E_0) \geq N(E_0 - \delta) - \frac{2\krm\, (1 + E_0)^2}{\delta}.~
\end{equation}

Let $\omega$ be the modulus of continuity of $N$. Then, by (\ref{eq:bk}) for any $E \leq E_0$ and $\delta \leq 1/2$
\[
|N (E) - N(E + \delta)| \leq \frac{C(d, E_0, A)}{\log_+^{\kappa_d}\frac{1}{\delta}}.~
\]
Thus, choosing $\delta = \frac{C (d, E_0, A)\, \krm}{\omega(C (d, E_0, A)\, \krm)}$, we obtain
\[
\sup_{E \leq E_0}|\tilde N(E) - N(E)| \leq \frac{\tilde C (d, E_0, A)}{\log_+^{\kappa_d}\frac{1}{\krm}},~
\]
therefore completing the proof.

\qed
\medskip

\textit{Acknowledgements:} I am grateful to Peter Hislop and Christoph Marx for helpful correspondence. I would like to thank Abel Klein for helpful comments. I would also like to thank Sasha Sodin for pleasant and useful discussions pertaining to the Ky Fan inequalities.

\end{document}